\documentclass[12pt]{article}

\usepackage{amsmath}
\usepackage[english]{babel}
\usepackage[ansinew]{inputenc}
\usepackage{epsfig}
\usepackage{color,graphicx}
\usepackage{wrapfig}
\usepackage{amssymb,amsthm,amsmath,booktabs}
\usepackage{dsfont}
\usepackage{natbib}
\usepackage{abstract}
\usepackage{ragged2e}

\usepackage{enumerate}

\renewcommand{\baselinestretch}{1}

\newtheorem{lemma}{Lemma}

\newtheorem*{definition*}{Definition}
\newtheorem*{thm*}{Theorem}
\newtheorem*{proposition*}{Proposition}

\newcommand{\bfnu}{\mbox{\boldmath $\nu$}}
\newcommand{\bftheta}{\mbox{\boldmath $\theta$}}
\newcommand{\bfTheta}{\mbox{\boldmath $\Theta$}}

\newcommand{\bbeta}{\mbox{\boldmath $\beta$}}
\newcommand{\bgamma}{\mbox{\boldmath $\gamma$}}
\newcommand{\bSigma}{\mbox{\boldmath $\Sigma$}}
\newcommand{\bmu}{\mbox{\boldmath $\mu$}}
\newcommand{\brho}{\mbox{\boldmath $\rho$}}
\newcommand{\p}{\mathbf{p}}
\newcommand{\X}{\mathbf{X}}
\newcommand{\y}{\mathbf{y}}
\newcommand{\Y}{\mathbf{Y}}
\newcommand{\U}{\mathbf{U}}
\renewcommand{\u}{\mathbf{u}}
\newcommand{\Z}{\mathbf{Z}}
\newcommand{\z}{\mathbf{z}}
\newcommand{\I}{\mathbf{I}}

\def\ds{\displaystyle}

\begin{document}

\title{Robust Bayesian model selection for heavy-tailed linear regression using finite mixtures}

\author{\bf{F. B. Gon\c{c}alves$^{a}$, M. O. Prates$^a$, V. H. Lachos$^b$}}

\maketitle

\begin{center}
{\footnotesize $^a$ Departamento de Estat\'{\i}stica, Universidade Federal de Minas Gerais, Brazil\\
$^b$ Departamento de Estat\'{\i}stica, Universidade Estadual de Campinas, Brazil}
\end{center}

\footnotetext[1]{Address: Av. Ant\^{o}nio Carlos, 6627 - DEST/ICEx/UFMG - Belo Horizonte, Minas Gerais, 31270-901, Brazil. E-mail: fbgoncalves@est.ufmg.br}

\maketitle

\begin{abstract}
In this paper we present a novel methodology to perform Bayesian
model selection in linear models with heavy-tailed distributions.
We consider a finite mixture of distributions to model
a latent variable where each component of the mixture corresponds to
one possible model within the symmetrical class of normal
independent distributions. Naturally, the Gaussian model is one of
the possibilities. This allows for a simultaneous analysis based on the
posterior probability of each model. Inference is performed via
Markov chain Monte Carlo - a Gibbs sampler with Metropolis-Hastings
steps for a class of parameters. Simulated examples highlight the
advantages of this approach compared to a segregated analysis based
on arbitrarily chosen model selection criteria. Examples with real data are
presented and an extension to censored linear regression is introduced
and discussed.
\end{abstract}


\section{Introduction}\label{secint}

Statistical practitioners generally use model selection
criteria in order to select a best model in different
applications. However, model selection has been shown not
to be an easy task and each criterion performs better under
different situations. For more complex models, it is not clear which
criterion is preferable \citep{Carlin:2006, Chen:2006, Gelman}.
Recently, \citet{Gelman} studied and
compared different model criteria and concluded that ``The current
state of the art of measurement of predictive model fit remains
unsatisfying''.
From their study it is clear that different criteria (AIC, DIC, WAIC)
fail in selecting the most adequate model under a variety of
circumstances. For example, settings with strong prior information or
when the posterior distribution is not well summarized by its mean
or in a spatial or network setup \citep[for more details, see][]{Gelman}.
Other authors also comment about the model selection
problem, e.g.,
``In summary, model choice is to Bayesians what multiple comparisons
is to frequentists: a really hard problem for which there exist several
potential solutions, but no consensus choice''\citep{Carlin:2006},
``we saw that no single measure is dominant in all three cases. The L
measure performed better when the true model becomes more complex
and BIC performed better when the true model is more parsimonious.''\citep{Chen:2006}.

Under the Bayesian paradigm, a more robust and elegant solution is available, at least in theory, by considering one ``full model" that
embeds all the individual models of interest. More specifically, this means that a multinomial r.v. with each category corresponding to
one of the individual models is specified. This way, model selection may be performed based on the posterior distribution of this r.v.
i.e., the posterior probability of each model. Nevertheless, this approach may be challenging in some cases, specially when the individual
models have different dimensions and distinct parameter. Available solutions may need to rely on complicated reversible jump MCMC algorithms
and, therefore, other model criteria may be preferred.

A simple and generally efficient solution may be obtained when a mixture distribution can be adopted for one of the model's component (parameter or latent variable) in a way that each mixture component corresponds to one of the individual models \citep[see, for example,][]{flavio,G&M}. This would typically lead to a simple and efficient solution, allowing for model selection to be based on the models' posterior probability.

We consider a model selection problem concerning the specification of the error distribution in linear regression models.
In particular, we consider different heavy-tailed distributions and the traditional Gaussian specification.
Existing solutions use model selection criteria arbitrarily chosen \citep[see][]{lachos, basso2010robust,
cabral2012multivariate} and, therefore, motivates the development of a more robust methodology.

The distributions of random errors and other random variables are routinely assumed to be
Gaussian. However, the normality assumption is doubtful and lacks
robustness especially when the data contain outliers or show a
significant violation of normality. Thus, previous works have shown
the importance of considering more general structures than the
Gaussian distribution for this component such as heavy-tailed
distributions \citep{Fern:1999,Galea2003,Rosa:2003,Galea:2005,garay2015}. These structures provide appealing robust and adaptable models,
for example, the Student-t linear mixed model presented by
\cite{pinheiro2001efficient}, who showed that it performed well
in the presence of outliers. Furthermore, the scale mixtures of
normal (SMN) distributions have also been applied into a wide
variety of regression models \citep[see][]{Lange93,
osorio2007assessment, lachos2011estimation}. It is one of the
most important subclasses of the elliptical symmetric distributions.
The SMN distribution class contains many heavier-than-normal tailed
members, such as Student-t, Slash, power exponential, and
contaminated normal. Recently, \cite{lin2013estimation} \citep[see
also][]{ lachos2011estimation} investigated the inference of a
measurement error model under the SMN distributions and demonstrated
its robustness against outliers through extensive simulations.

As defined by \citet{Andrews74}, a continuous random variable $Y$
has a SMN distribution if it can be expressed as follows
\begin{equation}
Y=\mu+\kappa^{1/2}(U)W,\nonumber
\end{equation}
where $\mu$ is a location parameter,  $W$ is a normal random
variable with zero mean and variance $\sigma^2$,  $\kappa(U)$ is a
positive weight function, $U$ is a mixing positive random variable
with density $h(.\mid \bfnu)$ and $\bfnu$  is a scalar or parameter
vector indexing the distribution of $U$. As in \citet{Lange93} and
\citet{Choy08},  we restrict our attention to the case
where $\kappa(U)=1/U$, that is, the normal independent (NI) class
of distributions. Thus, $Y\mid U=u\sim
\mathcal{N}(\mu,u^{-1}\sigma^2)$ and the marginal pdf of $Y$ is given by
\begin{equation}
f(y\mid\mu,\sigma^2,\bfnu)=\int^{\infty}_0 {\phi((y-\mu)/\sqrt{u^{-1} \sigma^2})}h(u\mid\bfnu)du.\label{SMNdef}
\end{equation}
Note that when $U=1$, we retrieve the normal distribution.
Following the steps of \citet{basso2010robust}, we have the following properties for the SMN family
\begin{itemize}
 \item[a)] If $E[\kappa^{1/2}(U)] < \infty$, then $E[Y] = \mu$, \\
 \item[b)] If $E[\kappa(U)] < \infty$, then $Var[Y] = \sigma^2 k_2$,\\
 \item[c)] If $E[\kappa^2(U)] < \infty$, then the excess kurtosis coefficient is given by
                $$\gamma_2 = \frac{E[Y - E[Y]]^4}{(Var[Y])^{2}}-3 = \frac{3k_4}{k_2^2 } - 3, $$
where $k_{m} = E[\kappa^{m/2}(U)]$.
\end{itemize}

Apart from the normal model, we explore two different types of
heavy-tailed densities based on the choice of $h(.\mid\bfnu)$.
\begin{itemize}
\item[$\bullet$]{\it   The Student-t distribution}, $Y\sim
\mathcal{T}(\mu,\sigma^2,\nu_t)$.\\
The use of the Student-t distribution as an alternative robust model
to the normal distribution has frequently been suggested in the
literature \citep{Lange89}. For the Student-t distribution with
location $\mu$, scale $\sigma$ and degrees of freedom $\nu_t$, the
pdf can be expressed as
\begin{equation}
f(y\mid \mu ,\sigma,\nu_t)=\int^{\infty}_{0} \phi((y-\mu)/\sqrt{u^{-1} \sigma^2}) f_{\mathcal{G}}( u \mid
\frac{\nu_t}{2},\frac{\nu_t}{2})du,\nonumber
\end{equation}
where $f_{\mathcal{G}}(. \mid a,b)$ is the Gamma density function with shape and rate parameters given by $a$ and $b$, respectively. That
is, $Y \sim \mathcal{T}_p(\mu,\sigma^2,\nu_t)$ is equivalent to the
following hierarchical form:
\begin{equation}
 Y\mid \mu,\sigma^{2}, \nu_t, u \sim \mathcal{N} \left( \mu,u^{-1}\sigma^{2}\right),\hspace{0.8 cm}\
 U \mid \nu \sim \mathcal{G}(\nu_t/2,\nu_t/2).\nonumber
\end{equation}
For the Student-t distribution we have that
$$k_{m} = \left(\frac{\nu_t}{2}\right)^{\frac{m}{2}} \frac{\Gamma(\frac{\nu_t-m}{2})}{\Gamma(\frac{\nu_t}{2})},$$
therefore, the Student-t has variance $\sigma^2\frac{\nu_t}{\nu_t-2}$, for $\nu_t>2$, and excess kurtosis $\frac{6}{\nu_t-4}$, for $\nu_t>4$.
\item[$\bullet$]{\it The Slash distribution,  $Y\sim
\mathcal{S}(\mu,\sigma^2,\nu_s)$.} \\ This distribution
presents heavier tails than those of the normal distribution and it
includes the normal case when $\nu_s \uparrow \infty$. Its pdf is
given by
\begin{equation}
f(y\mid \mu ,\sigma,\nu_s)=\nu_s\int^1_0{u^{\nu_s-1}\phi((y-\mu)/\sqrt{u^{-1} \sigma^2})}d u.\nonumber
\end{equation}
Thus, the Slash distribution is equivalent to the following
hierarchical form:
\begin{equation}
 Y \mid \mu,\sigma^{2}, \nu_s, u \sim \mathcal{N} \left( \mu,u^{-1}\sigma^{2}\right), \hspace{0.8 cm}U \mid \nu_s \sim \mathcal{B}(\nu_s,1),\nonumber\\ \nonumber\\ \nonumber
\end{equation}
where $\mathcal{B}(.,.)$ denotes the beta distribution.
For the Slash distribution we have that
$$k_{m} = \frac{2 \nu_s}{2 \nu_s - m},$$
therefore, the Slash has variance $\sigma^2\frac{\nu_s}{\nu_s-1}$, for $\nu_s>1$, and excess kurtosis
$\frac{3}{\nu_s(\nu_s-2)}$, for $\nu_s>2$.
\end{itemize}

The SMN formulation described above is used in a linear regression approach by taking $\mu= \X_i \bbeta$ where $\bbeta$ is the vector of
coefficients and $\X$ is the design matrix.

The aim of this paper is to propose a general formulation to perform Bayesian model
selection for heavy-tailed linear regression models in a simultaneous setup.
That is achieved by specifying a full model which includes the space of all individual models under consideration - specified using the SMN approach described above. This way, the model selection criterion can be based on the posterior probability of each model. A mixture distribution is adopted to one of the full model's variable, with each component of the mixture referring to one of the individual models.
This approach has two main advantages when compared to an ordinary analysis
where each model is fitted separately and some model selection criterion is used.
Firstly, there is a significant gain in the computational cost since we eliminate the need to fit all the individual models separately.
Secondly, the proposed model selection criterion is fully based on the Bayesian Paradigm, meaning that the model choice is based on
the posterior probability of each model. This is more robust when compared to some other
arbitrarily chosen model selection criteria such as DIC, EAIC, EBIC \citep{Spiegelhalter},
CPO \citep{Geisser} WAIC \citep{Watanabe}.
The examples presented in the paper are meant to provide empirical evidence for this argument.
The posterior distribution of the unknown quantities has a significant level of
complexity which motivates the derivation of a MCMC algorithm to obtain a sample
from this distribution.

This paper is organised as follows: Section~\ref{secmodel} presents the general model;
Section~\ref{secinference} presents a MCMC algorithm to
make inference for the proposed model; a variety of simulated examples are presented
in Section~\ref{secsimul} and the analysis of two real data sets is shown in
Section~\ref{secreal}. Finally, Section~\ref{secconc} discusses some extensions of
the proposed methodology.

\section{Linear regression model with heavy-tailed mixture structured errors}\label{secmodel}

Model selection is an important and complex problem in statistical
analysis and the Bayesian approach is particularly appealing to
solve it. In particular, the use of mixtures is a nice way to pose
and solve the problem, whenever possible. It allows for an analysis
where all models are considered and compared in a simultaneous setup
without the need of complicated reversible jump MCMC algorithms.
Note that, from (\ref{SMNdef}), each model is determined by the
distribution of the scale factor $u$, which suggests
that a mixture distribution could be used for this latent variable.
We present a general finite mixture model framework capable of capturing
different behavior of the response and indicate which individual
distribution is preferred.

\subsection{The model}
\label{s:model}

Define the $n$-dimensional response vector $\Y$, the $n\times q$
design matrix $\X$, the $q$-dimensional coefficient vector $\bbeta$
and two $K$-dimensional vectors
$\bgamma=(\gamma_1\;\ldots\;\gamma_K)'$ and
$\p=(p_1\;\ldots\;p_K)'$. Finally, let $diag(\u^{-1})$ be a
$n$-dimensional diagonal matrix with $i$-th diagonal $u_{i}^{-1}$,
$\;i=1,\ldots,n$. We propose the following general model:
\begin{eqnarray}
\ds (\Y|Z_j=1, \U = \u) &\sim& \mathcal{N}\left(\X \bbeta,\sigma^2\gamma_jdiag(\u^{-1})\right) \label{modeleq1}\\
   (U_i|Z_j=1)&\stackrel{iid}{\sim}&F_j(\nu_j), \;i=1,\ldots,n, \label{modeleq3} \\ 
   \Z&\sim& Mult(1,p_1,\ldots,p_K) \label{modeleq2} \\
   \gamma_j&=&g_j(\nu_j),\;j=1,\ldots,K, \label{modeleq4}
\end{eqnarray}
where $Mult$ is the Multinomial distribution and each $F_j$ represents a positive distribution controlled by parameter(s) $\nu_j$, which may need to be truncated to guarantee that $Y_i$ has finite variance under each $F_j$.

The particular structure chosen for the variance in (\ref{modeleq1}) was thought of
so that, for each $j$, the variance of the model is the same - $\sigma^2$. This is achieved through specific choices for the
functions $\gamma_j$ and  allows us to treat $\sigma^2$ as a common parameter to all of the individual models.
Otherwise we would need one scale parameter for each model. Therefore, in our approach,
since we have a common $\mu$ and $\sigma^2$, the models will mainly differ from each other in terms of tail behavior which will
favor the model selection procedure.

Note that each component from the mixture distribution of $u_i$ corresponds to one of the
models being considered. Model selection is made through the posterior distribution of $\Z$.
A subtle but important point here is the fact that there is no $i$ index for $Z_j$. This means that we assume that all the observations come from the same model, which poses the inference problem in the model selection framework.

Another advantage of the simultaneous approach is that it allows the use of Bayesian model averaging \citep[see][]{RMV}. This is particularly
useful in cases where more than one model have a significant posterior probability. Note that the models we consider can be quite similar in some situations - specially for higher values of the degrees of freedom (df) parameters.

\subsection{Prior distributions}

The Bayesian model is fully specified by (\ref{modeleq1})-(\ref{modeleq4}) and the prior
distribution for the parameter vector $\bftheta=(\bbeta,\sigma^2,\p,\bfnu)$, for
$\bfnu=(\nu_1,\ldots,\nu_K)$. Due to the complexity of the proposed model, the prior
distribution plays an important role on the model identifiability and selection process
and, for that reason, needs to be carefully specified.

Prior specification firstly assumes independence among all the components of $\bftheta$.
Secondly, standard priors $\bbeta\sim\mathcal{N}_q(\bmu_0,\tau_{0}^2\I_q)$ and
$\sigma^2\sim\mathcal{I}\mathcal{G}(a_0,b_0)$ are adopted.

The prior distributions of the tail behavior parameters $\bfnu$ require special attention.
This type of parameter is known to be hard to estimate \citep[see][]{steel} and the most
promising solutions found in the literature tackle the problem through special choices of
prior distributions \citep[see][]{fonseca}. Recently, \citet{rue} proposed a general family
of prior distributions for flexibility parameters which includes tail behavior parameters.

In this paper we adopt the penalised complexity priors (PC priors) from \citet{rue}.
In a simple way, the PC priors have as main principle to prefer a simpler model and penalise the more complex
one. To do so, the Kullback-Leibler divergence (KLD) \citep{kullback} is used to define a measure
of information loss when a simpler model $h$ is used to approximate a more flexible model
$f(\cdot|\nu_j)$. The measure $d(f||h)(\nu_j)= d(\nu_j) = \sqrt{2KLD(f||h)}$ is defined to be a measure of complexity
of model $f(\cdot|\nu_j)$ in comparison to $h$. Further, a density function $\pi(d(\nu_j)) = \lambda \exp(-\lambda d(\nu_j))$
is set for the measure $d(\nu_j)$. Finally, the prior distribution of $\nu_j$ is given by
$$\pi(\nu_j) = \lambda \exp(-\lambda d(\nu_j)) \left| \frac{\partial d(\nu_j)}{\partial \nu_j}\right|
\; j = 1, \ldots, K.$$
\citet{martins} showed that in a practical way, for the Student-t regression model, the PC prior
can behave very similar to the Jeffrey's priors constructed by \citet{fonseca}. Another
interesting practical usage of this prior is that the selection of an appropriate $\lambda$ is done
by allowing the researcher to control the prior tail behavior of the model. For example, for the Student-t
distribution the user must select $\nu^\star$ and $\xi$ such that $P(\nu_j < \nu^\star) = \xi$, in other words,
how much mass probability $\xi$ is assigned to $\nu_j \in (2,\nu^\star)$
(where $j$ defines the $F_j$ distribution such that the response follows a Student-t distribution).
Clearly, the same procedure applies for any other distribution in the NI family that has a flexibility
parameter. For more details on the PC priors see \citet{rue}.

The prior distribution for $\p$ also requires special attention. Note that even in the
extreme (unrealistic) case where $\Z$ is observed, it does not provide much information about $\p$,
in fact, it is equivalent to the information contained in a sample of size one from a $Mult(1,p_1,\ldots,p_K)$
distribution. The fact that $\Z$ is unknown aggravates the problem. A simple and practical way to
understand the consequences of this is given by the following lemma, which is a generalisation of Lemma 1
from \citet{flavio} where, to the best of our knowledge, this problem was firstly encountered.
\begin{lemma}\label{pprior}
For a prior distribution $\p\sim Dir(\alpha_1,\ldots,\alpha_K)$, the posterior mean of $p_j,\;\forall j$,
is restricted to the interval $\ds \left(\frac{\alpha_j}{1+\sum_{k=1}^K\alpha_k},\frac{\alpha_j+1}{1+\sum_{k=1}^K\alpha_k}\right)$.
\end{lemma}
\begin{proof}
See Appendix~\ref{a:proof}.
\end{proof}
For example, if $\alpha_j=1,\;\forall j$, then $\mathbb{E}[p_j|y]\in(1/(K+1),2/(K+1))$.
This result indicates that the estimation of $\Z$ may be compromised by unreasonable choices of the $\alpha_j$'s.

A reasonable solution for this problem is to use a Dirichlet prior distribution with
parameters (much) smaller than 1, which makes it sparse. It is important, though, to
choose reasonable values for the $\alpha_j$'s, in the light of Lemma~\ref{pprior}.
\citet{flavio} claim that $\alpha_j=0.01,\;\forall j$ leads to good results and,
in the cases where prior information is available, some of the $\alpha_j$'s
may be increased accordingly.

\section{Bayesian Inference}\label{secinference}

We derive a MCMC algorithm considering the three most common
choices in the NI family - Normal, Student-t, Slash.
Nevertheless, based on the formulation presented in
Section~\ref{s:model}, including other possibilities is straightforward.
One should be careful, however, as it may lead to serious identifiability issues due to
similarities among the individual models. The model is given by:

\begin{eqnarray}
\ds  (\Y|Z_j=1) &\sim& \mathcal{N}\left(\X \bbeta,\sigma^2\gamma_jdiag(\u^{-1})\right) \label{modeleq1a}\\
   \Z&\sim& Mult(1,p_1,p_2,p_3) \label{modeleq2a} \\
   U_i&\stackrel{iid}{\sim}&\left\{
              \begin{array}{ll}
                \delta_1, & \mbox{if } Z_1=1 \\
                \mathcal{G}\left(\nu_t/2,\nu_t/2\right), & \mbox{if } Z_2=1, \; i = 1,\ldots,n, \\
                \mathcal{B}(\nu_s,1), & \mbox{if } Z_3=1
              \end{array}
            \right. \label{modeleq3a}
     \\
   \gamma_j&=&\left\{
              \begin{array}{ll}
                1, & \;\;\;\; \mbox{if } Z_1=1 \\
                (\nu_t-2)/\nu_t, & \;\;\;\;\mbox{if } Z_2=1 \\
                (\nu_s-1)/\nu_s, & \;\;\;\;\mbox{if } Z_3=1,
              \end{array}
            \right. \label{modeleq4a}
\end{eqnarray}
where $\delta_1$ is a degenerate r.v. at 1 and $\mathcal{G}$ and $\mathcal{B}$ are the Gamma and Beta distributions, respectively.
We impose that $\nu_t>2$ and $\nu_s>1$ so that $Y_i$ has finite variance ($\sigma^2$) under each individual model.

Inference is performed via MCMC - a Gibbs sampling with Metropolis Hastings (MH) steps for
the degrees of freedom parameters. Details of the algorithm are presented below.

\subsection{MCMC}
\label{s:mcmc}

We choose the following blocking scheme for the Gibbs sampler:
\begin{equation}\label{GSblock}
(\p,\Z,\U)\;,\;\bbeta\;,\;\sigma^2\;,\;(\nu_t,\nu_s).
\end{equation}
This blocking scheme minimises the number of blocks among the algorithms with only one
MH step (which is inevitable for the df parameters). The minimum number of blocks
reduces the correlation among the components, which speeds the convergence of the chain.
Moreover, the most important and difficult step is the one that samples from $(\p,\Z,\U)$
and sampling directly from its full conditional also favors the convergence properties of the chain.

The full conditional densities of (\ref{GSblock}) are all derived from the joint density of
all random components of the model.
\begin{eqnarray}\label{fpost}
\ds&&\pi(\Y,\bbeta,\sigma^2,\p,\Z,\U,\gamma,\nu_t,\nu_s|\X)\propto \nonumber \\
&&\pi(\Y|\bbeta,\sigma^2,\Z,\U,\gamma,\X)\pi(\U|\Z,\nu_t,\nu_s)\pi(\Z|\p)\pi(\p)\pi(\nu_t)\pi(\nu_s)
\pi(\bbeta)\pi(\sigma^2).\nonumber\\
\end{eqnarray}
The first two terms on the right hand side of (\ref{fpost}) are
given in Section \ref{secint}, for each individual model ($Z_j$).
The remaining terms are given in Section \ref{secmodel}.

The full conditional distributions of $\bbeta$ and $\sigma^2$ are easily devised and given by:
\begin{eqnarray} \nonumber
\ds(\bbeta|\cdot) &\sim& \mathcal{N}_q\left(\bSigma_{\beta}\left((\tau_{0}^2\I_q)^{-1}\bmu_0+(\sqrt{\u}\odot \X)'(\sqrt{\u}\odot \y)/(\gamma_j\sigma^2)\right)\;,\;\bSigma_{\beta}\right) \label{fc1} \\ \nonumber
  (\sigma^2|\cdot) &\sim& \mathcal{I}\mathcal{G}\left(a_0+n/2\;,\;b_0+\sum_{i=1}^n\frac{u_i(y_i-\X_{i\cdot}\bbeta)^2}{2\gamma_j}\right),\label{fc2}
\end{eqnarray}
where $\ds \bSigma_{\beta}=\left((\tau_{0}^2\I_q)^{-1}+(\sqrt{\u}\odot \X)'(\sqrt{\u}\odot \X)/(\gamma_j\sigma^2)\right)^{-1}$, $\sqrt{\u}$
is the $n$-dimensional vector with entries $\sqrt{\u_i}$, $\odot$ is the Hadamard product which multiplies term by
term of matrices with the same dimension and $\I_q$ is the identity matrix with dimension $q$.

The df parameters are sampled in a MH step with the following transition distribution (at the $k$-th iteration):
\begin{eqnarray}
\ds q\left(\nu_{t}^{k},\nu_{s}^{k}\right)&=&q(\nu_{t}^{k})q(\nu_{s}^{k}) \label{fc3a} \\
q(\nu_{t}^{k})&=&\left((1-Z_2)\mathds{1}(\nu_{t}^{k}=\nu_{t}^{k-1})+Z_2f_{\mathcal{N}}(\nu_{t}^{k};\nu_{t}^{k-1},\tau_{t}^2)\right) \label{fc3b} \\
q(\nu_{s}^{k})&=&\left((1-Z_3)\mathds{1}(\nu_{s}^{k}=\nu_{s}^{k-1})+Z_3f_{\mathcal{N}}(\nu_{s}^{k};\nu_{s}^{k-1},\tau_{s}^2)\right),\label{fc3c}
\end{eqnarray}
where $f_{\mathcal{N}}(l;a,b)$ is the density of a normal distribution with mean $a$ and variance $b$
evaluated at $l$. The respective acceptance probability of a move is
\begin{equation}\label{MHap}
\ds\alpha(k-1\rightarrow k)=
\min\left\{1,Z_1+Z_2\frac{\pi(\nu_{t}^k|\cdot)}{\pi(\nu_{t}^{k-1}|\cdot)}+Z_3\frac{\pi(\nu_{s}^k|\cdot)}{\pi(\nu_{s}^{k-1}|\cdot)}\right\},
\end{equation}
where
\begin{eqnarray}\nonumber
\ds \pi(\nu_t|\cdot)&\propto&\pi(\U|Z_2=1,\nu_t)\pi(\nu_t) \\ \nonumber
\pi(\nu_s|\cdot)&\propto&\pi(\U|Z_3=1,\nu_s)\pi(\nu_s).
\end{eqnarray}
This result is obtained by adopting the following dominating measure for both
the numerator and the denominator of the acceptance probability: $\mathbb{L}^2\otimes \mathbb{L}\otimes m$
if $Z_1=0$ and $\mathbb{L}^2\otimes m^2$ if $Z_1=1$, where $m$ is the counting measure
and $\mathbb{L}^d$ is the $d$-dimensional Lebesgue measure. 
The detailed balance along with the fact that chain is irreducible, makes this a valid MH algorithm \citep[see][]{tierney}.

Note that, once we have the output of the chain, estimates of the df parameters will be based on samples of $(\nu_t|Z_2=1)$ and $(\nu_s|Z_3=1)$, which justifies the transition distributions in (\ref{fc3a})-(\ref{fc3c}).

From (\ref{fpost}), the full conditional density of $(\p,\Z,\U)$ is
\begin{eqnarray}
\ds \pi(\U,\Z,\p|\cdot)&\propto&\pi(\y|\bbeta,\sigma^2,Z,U,\gamma,\X)\left[\prod_{i=1}^n{\pi(U_i|\Z,\nu_t,\nu_s)}\right]\pi(Z|p)\pi(\p) \nonumber \\
&\propto&\left[\prod_{i=1}^n{\pi(U_i|\cdot)}\right](r_1p_1)^{Z_1}(r_2p_2)^{Z_2}(r_3p_3)^{Z_3}\pi(\p). \nonumber
\end{eqnarray}
Defining $\ds w=\sum_{j=1}^3r_jp_j$ and $w_j=r_jp_j/w$, for $j=1,2,3$, we get
\begin{equation}\label{fc4b}
\ds \pi(\U,\Z,\p|\cdot)\propto\left[\prod_{i=1}^n{\pi(U_i|\cdot)}\right](w_1)^{Z_1}(w_2)^{Z_2}(w_3)^{Z_3}w\pi(\p).
\end{equation}
We can sample from (\ref{fc4b}) using the following algorithm.\\
\\{\scriptsize
\begin{tabular}[!]{|l|}
\hline\\
\parbox[!]{10cm}{
\texttt{
\begin{enumerate}
\item Simulate $\p$ from a density $\pi^*(\p) \propto w \pi(\p)$;
\item Simulate $\Z \sim Mult(1,w_1,w_2,w_3)$;
\item Simulate $U_i$ from the density $\pi(U_i|\cdot),\;\forall i$;
\item OUTPUT $(\u,\z,\p)$.
\end{enumerate}
}}\\  \hline
\end{tabular}}\\

Steps 2 and 3 are straightforward once we have that:
\begin{eqnarray} \nonumber
\ds r_1 &=& \prod_{i=1}^n\exp\left(-\frac{1}{2\gamma_1\sigma^2}\tilde{y}_{i}^2\right); \\ \nonumber
\ds r_2&=&\frac{\left(\frac{\nu_t-2}{\nu_t}\right)^{-n/2}\left(\nu_t/2\right)^{n\nu_t/2}\left(\Gamma\left(\frac{\nu_t+1}{2}\right)\right)^n}{\left(\Gamma\left(\frac{\nu_t}{2}\right)\right)^n\prod_{i-1}^n\left(\frac{\tilde{y}_{i}^2}{2\gamma_2\sigma^2}+\frac{\nu_t}{2}\right)^{(\nu_t+1)/2}};\\ \nonumber
\ds r_3&=&\left(\frac{\nu_s-1}{\nu_s}\right)^{-n/2}\left(\frac{\Gamma(\nu_s+1)}{\Gamma(\nu_s)}\Gamma(\nu_s+1/2)\right)^n\prod_{i=1}^n\left[\frac{F_{\mathcal{G}}\left(1;\nu_s+1,\frac{\tilde{y}_{i}^2}{2\gamma_3\sigma^2}\right)}{\left(\frac{\tilde{y}_{i}^2}{2\gamma_3\sigma^2}\right)^{\nu_s+1/2}}\right],
\end{eqnarray}
where $\tilde{y}_{i}=y_i-\X_{i\cdot}\bbeta$ and
$\ds F_{\mathcal{G}}(x;a,b)$ is the distribution function of a Gamma distribution
with parameters $(a,b)$ evaluated at $x$. Moreover,
\begin{eqnarray} \nonumber
\ds (U_i|Z_1=1,\cdot)&\sim&\delta_1; \\ \nonumber
\ds (U_i|Z_2=1,\cdot)&\sim&\mathcal{G}\left((\nu_t+1)/2\;,\;\tilde{y}_{i}^2/(2\gamma_2\sigma^2)+\nu_t/2\right); \\ \nonumber
\ds (U_i|Z_3=1,\cdot)&\sim&\mathcal{G}_{[0,1]}\left(\nu_s+1\;,\;\tilde{y}_{i}^2/(2\gamma_3\sigma^2)\right),
\end{eqnarray}
where $\mathcal{G}_{[0,1]}$ is a truncated Gamma distribution in $[0,1]$.

Step 1 is performed via rejection sampling (RS) proposing from the prior $\pi(\p)$ and
accepting with probability $\frac{w}{\max_j\{r_j\}}$. Simulated studies indicated that
the algorithm is computationally efficient.

Monte Carlo estimates of the posterior distribution of $\Z$ (denoted by $\brho$), i.e. the models' posterior probabilities,
based on a sample of size $M$, are given by
\begin{equation}
\ds \hat{\rho}_j=\widehat{P(Z_j=1|y)}=\frac{1}{M}\sum_{m=1}^M\mathds{1}(Z_{j}^{(m)}=1),\;\;j=1,2,3.\nonumber
\end{equation}

\subsection{Practical implementation}

The MCMC algorithm described in the previous section requires special attention
to some aspects to guarantee its efficiency.

An indispensable strategy consists of warming up the chain inside
each of the heavy-tailed models (Student-t and Slash). It
contributes in several ways to the efficiency of the algorithm.

Firstly, it contributes to the mixing of the chain among the different models.
If the chain starts at arbitrary values for the df parameters, it may move to
high posterior density values for one of them while the other is still at a
low posterior density value. This will make moves from the former model to the
latter very unlike, jeopardising the convergence. More specifically, one may take
the sample mean of the df parameters from their respective warm-up chains,
after discarding a burn-in, as the starting values for the full chain.

Secondly, the warm-up chains will achieve or approach local convergence
(inside each model). This will significantly speed the convergence of the full chain,
which will have as main purpose the convergence of the $\Z$ coordinate.

Finally, the warm-up chains are a good opportunity to tune the MH steps of the df
parameters. Given the unidimensional nature of the step and the random walk
structure, the acceptance rates should be around 0,44 \citep[see][]{roberts}.

\subsection{Prediction}

An often common step in any regression analysis is prediction for a new configuration $\X_{n+1}$ of the covariates. This procedure is straightforward in a MCMC context where a sample from the posterior predictive distribution of $Y_{n+1}$ can be obtained by adding two simple steps at each iteration of the Gibbs sampler after the burn-in.

Let $\left(\Z^{(m)},\bbeta^{(m)},{\sigma^2}^{(m)},\bgamma^{(m)},\nu_{t}^{(m)},\nu_{s}^{(m)}\right)$ be the state of the chain at the $m$-th iteration after the burn-in. Then, for each $m=1,2,\ldots$, firstly sample
$(u_{n+1}^{(m)}|\Z^{(m)},\nu_{t}^{(m)},\nu_{s}^{(m)})$ from (\ref{modeleq3a}) and finally sample
\begin{equation}\label{ppdY}
Y_{n+1}^{(m)}\sim\mathcal{N}\left(\X_{n+1}\bbeta^{(m)},{\sigma^2}^{(m)}({\Z'}^{(m)}\bgamma^{(m)})(u_{n+1}^{(m)})^{-1}\right),
\end{equation}
where ${\Z'}^{(m)}$ is a row vector and $\bgamma^{(m)}$ is a column vector.

One can also consider the posterior predictive distribution of $Y_{n+1}$ under one particular model, for example, the one with the highest posterior probability. In that case, it is enough to consider the sub-sample of the sample above corresponding to the chosen model.

\section{Simulated examples}
\label{secsimul}
In this section we introduce synthetic data examples to
better understand the properties of the proposed methodology.
Our goal are two fold: to provide strong empirical evidence that, (1) as long as information is
available, the true model is selected using the proposed methodology and (2) this selects the correct
or more adequate model more often than some traditional criteria.

We firstly present a study to see how well the proposed methodology
correctly identifies the true model. A second study shows the performance of the
criteria when the model has correlated covariates, which may cause problems in
the estimate of the fixed effects and variance parameter. Finally, a third
synthetic data set is generated from a residual mixture model
to investigated if the model that better approximates the true mixture distribution is chosen.

\subsection{Study I}
\label{secsimI}

In this study, data is generate data from one of the proposed distributions: Normal, Student-t
and Slash. We consider an intercept and two covariates, i.e. $\X_{i\cdot} = (1, X_{i1}, X_{i2})$, where $X_{i1}$ is a
standard Normal random variable, $X_{i2}$ is a Bernoulli
random variable with parameter $0.5$ and $i = 1, \ldots, n$.
The regression coefficients are $\bbeta^\top = (1,2,-2)$.
Finally for all models, the variance $\sigma^2$ is set to $1$.
The synthetic data were generated from each of the
following distributions:
\begin{enumerate}
 \item Normal;
 \item Student-t with degrees of freedom $\nu_t = 15$ and $\nu_t = 3$;
 \item Slash with degrees of freedom $\nu_s = 3.36$ and $\nu_s = 1.25$.
\end{enumerate}
Different sample sizes $n$ are also considered - $100$, $500$, $1000$
and $5000$, giving a total of 20 scenarios.
The degrees of freedom for the Slash were chosen to minimise the
Kullback-Leibler divergence between the Student-t with $\nu_t = 15$
and $\nu_t = 3$, respectively.

For each simulated scenario a Markov chain runs for $110k$ iterations,
with a burn-in of $10k$ giving a total posterior chain of $100k$
iterations. Convergence is checked using Geweke's criterion \citep{Gewe:eval:1992}
since we only ran one chain. The same chain size, burn-in period and convergence
verification were performed for all the examples in the paper.
Notice that the parametrisation adopted allows some parameters to be estimated using the whole
chain, independently of the model that
is visited in each iteration. This favors the chain convergence and the MC variance of the estimates of
($\bbeta, \sigma^2$).

The summary posterior results of one run are presented in Table~\ref{t:res}.
They show that as the sample size increases the
proposed methodology selects the correct model. Moreover, in the case where
data is generated from the Normal distribution, not only the correct model is correctly
chosen in all but one case, but also the estimated degrees of freedom of the Student-t
and of the Slash distributions are high - making these distributions similar to Gaussian.
Another important feature presented in the Table~\ref{t:res} is that the degrees of
freedom parameter of the generating model is well estimated.
For the non-generating model, the degrees of freedom parameter is reasonably estimated, in the sense of
making the respective model as close as possible to the true one.
For example, when the data is generated from the Student-t with $\nu_t = 15$ the
$\nu_s$ is estimated close to $3.36$, which is the value that minimises the
Kullback-Leibler divergence between the two distributions.
Table~\ref{t:res} also emphasises that, for small sample sizes $n=100$ or $n=500$,
there is not enough information about the tail behavior to clearly distinguish among the
models.

\renewcommand{\baselinestretch}{1}

\begin{table}[htb]
  \caption{\label{t:res} Results for Study I. Estimates (posterior mean) refer to one of the 50 replications.}
  \centering
  \resizebox{\textwidth}{!}{%
  \begin{tabular}{cc cc cc}
    \toprule
    Model & sample size & $\beta^\top = (1, 2, -2$) & $\sigma^2 = 1$ & ($\nu_t$, $\nu_s$) & $\brho = (\rho_1,\rho_2,\rho_3)$ \\
    \midrule
              & 100     & (1.153, 1.991, -2.305) & 1.121  & (10.62, 2.09) & (0.103, 0.537, 0.360) \\
              & 500     & (0.997, 2.047, -2.065) & 0.979  & (29.47, 4.43) & (0.882, 0.073, 0.045) \\
    Normal    & 1000    & (1.004, 1.981, -1.986) & 0.999  & (31.20, 4.45) & (0.644, 0.280, 0.076) \\
              & 5000    & (0.990, 1.979, -1.965) & 0.980  & (44.25, 5.32) & (0.749, 0.097, 0.154) \\
   \midrule
                            & 100     & (1.236, 1.829, -2.148) & 1.267  & (9.86, 1.86)  & (0.044, 0.439, 0.517) \\
                            & 500     & (1.074, 2.029, -2.042) & 1.038  & (28.64, 4.14) & (0.777, 0.151, 0.072) \\
    Student-t ($\nu_t = 15$)& 1000    & (1.012, 2.006, -1.991) & 0.982  & (21.24, 3.72) & (0.123, 0.609, 0.268) \\
                            & 5000    & (1.014, 2.000, -1.999) & 0.993  & (16.19, 3.19) & (0.000, 0.807, 0.193) \\
   \midrule
                            & 100     & (1.116, 1.865, -2.045) & 1.389  & (3.22, 1.22)  & (0.000, 0.371, 0.629) \\
                            & 500     & (0.978, 2.031, -1.923) & 1.244  & (3.36, 1.20)  & (0.000, 0.679, 0.321) \\
    Student-t ($\nu_t = 3$) & 1000    & (1.001, 2.005, -1.959) & 0.861  & (3.30, 1.25)  & (0.000, 0.990, 0.010) \\
                            & 5000    & (1.024, 2.007, -2.035) & 1.029  & (2.95, $-$)   & (0.000, 1.000, 0.000) \\
   \midrule
                           & 100     & (0.968, 2.097, -1.902) & 1.049  & (17.37, 2.76)  & (0.369, 0.357, 0.274) \\
                           & 500     & (0.976, 2.003, -2.039) & 0.963  & (19.90, 3.30)  & (0.167, 0.450, 0.383) \\
    Slash ($\nu_s = 3.36$) & 1000    & (1.004, 1.997, -2.010) & 1.015  & (17.72, 3.22)  & (0.020, 0.626, 0.354) \\
                           & 5000    & (1.029, 2.000, -2.044) & 0.963  & (22.61, 3.65)  & (0.000, 0.230, 0.770) \\
   \midrule
                           & 100     & (1.012, 1.988, -1.957) & 0.454  & (18.04, 2.75) & (0.344, 0.367, 0.289) \\
                           & 500     & (1.033, 2.026, -2.015) & 0.904  & (3.91, 1.29)  & (0.000, 0.280, 0.720) \\
    Slash ($\nu_s = 1.25$) & 1000    & (1.012, 2.012, -2.040) & 0.839  & (3.93, 1.35)  & (0.000, 0.561, 0.439) \\
                           & 5000    & (1.017, 1.988, -2.011) & 0.863  & ($-$, 1.30)   & (0.000, 0.000, 1.000) \\
   \bottomrule
\end{tabular}}
\end{table}


To check the capability of the proposed methodology in selecting the correct model, we performed
a Monte Carlo study with $50$ replicates of each of the $20$ generation schemes.
Table~\ref{t:eqm} presents the Mean Square Error (MSE) of the posterior estimates of the model parameters.
The MSE of the $\nu$ parameters was calculated considering only the replications in which the true model was selected.
The MSE of $\rho$ is the mean square error between the posterior estimate of the true model's probability and 1.
From Table~\ref{t:eqm}
we can see that, even for small sample sizes, the MSE values of $\beta$ and $\sigma^2$
indicate a very good recovery of the true values. The MSE values
of $\nu_t$ include some large values for small samples sizes when there is not enough information to estimate
precisely the degrees of freedom. The difference in the magnitudes of $\nu_t$ and $\nu_s$ are explained by the
difference in the scale of those parameters. Finally, the last column of the table
shows the percentage of times that the correct model was selected (i.e. had the highest posterior probability).

\renewcommand{\baselinestretch}{1}

\begin{table}[htb]
  \caption{\label{t:eqm} Mean Square error estimates from the 50 replications for each parameter.
  The last column is the percentage of correct selection in the 50 replications. The $(\times10^{3})$
  means that the reported value is MSE$\times10^{3}$.}
  \centering
  \resizebox{\textwidth}{!}{%
  \begin{tabular}{cc cc cc c}
    \toprule
    Model & sample size & $\beta (\times10^{3})$ & $\sigma^2 (\times10^{3})$ & $\nu_t$ or $\nu_s$ & $\rho$ & *Pct SCM \\
    \midrule
              & 100     & (16.29, 10.15,  34.44) & 32.12  & - & 0.303 & $80\%$ \\
              & 500     &  (3.24,  1.84,   7.23) &  3.96  & - & 0.207 & $86\%$\\
    Normal    & 1000    &  (2.24,  0.75,   5.03) &  2.24  & - & 0.192 & $84\%$\\
              & 5000    &  (0.40,  0.24,   0.66) &  0.46  & - & 0.123 & $88\%$\\
    \midrule
                             & 100     & (26.04, 11.27,  44.80) & 72.04  & 15.852 & 0.456 & $30\%$ \\
                             & 500     &  (4.39,  2.42,   9.83) &  5.91  & 45.564 & 0.288 & $64\%$\\
    Student-t ($\nu_t = 15$) & 1000    &  (2.25,  0.82,   5.40) &  2.84  & 36.736 & 0.284 & $68\%$\\
                             & 5000    &  (0.35,  0.18,   0.67) &  0.46  & 11.580 & 0.131 & $80\%$\\
    \midrule
                             & 100     & (10.82,  6.26,  23.04) & 97.85  & 71.458 & 0.365 & $32\%$ \\
                             & 500     &  (2.81,  0.87,   4.82) & 83.92  & 0.580  & 0.226 & $62\%$\\
    Student-t ($\nu_t = 3$)  & 1000    &  (0.95,  0.55,   1.98) & 53.03  & 0.151  & 0.176 & $76\%$\\
                             & 5000    &  (0.16,  0.09,   0.42) &  5.79  & 0.022  & 0.000 & $100\%$\\
    \midrule
                             & 100     & (21.90,  10.93,  33.79) & 45.65  & 1.625 & 0.557 & $10\%$ \\
                             & 500     &  (2.71,   1.82,   5.89) &  5.27  & 0.417 & 0.447 & $32\%$\\
    Slash ($\nu_s = 3.36$)   & 1000    &  (1.49,   0.84,   3.35) &  3.38  & 0.185 & 0.412 & $40\%$\\
                             & 5000    &  (0.42,   0.19,   0.74) &  0.70  & 0.176 & 0.268 & $54\%$\\
    \midrule
                             & 100     &  (7.17,  5.37,  16.46) & 86.38  & 0.056 & 0.259 & $62\%$ \\
                             & 500     &  (2.46,  1.16,   3.94) & 40.40  & 0.021 & 0.193 & $72\%$\\
   Slash ($\nu_s = 1.25$)    & 1000    &  (0.95,  0.67,   1.58) & 49.71  & 0.006 & 0.200 & $64\%$\\
                             & 5000    &  (0.24,  0.09,   0.56) & 23.65  & 0.001 & 0.206 & $78\%$\\
   \bottomrule
\end{tabular}}
\end{table}

\subsection{Study II}
\label{secsimII}

This study investigates how the model selection procedure and the parameter estimation is affected in the presence of correlated covariates.
We generate data from a model with $e_i \sim \mathcal{T}(0,1,3)$, $\beta = (1,2,-2,1)$, $\X_{i} = (1, X_{i1}, X_{i2}, X_{i3})$ and
$X_{i3} = 2X_{2i} + \mathcal{N}(0,0.5)$ which induces an average correlation of $0.9$ between the two covariates.

We reproduce $50$ replicates of this scenario with different sample
sizes $n = 500$, $1000$, $2000$ and $5000$. Table~\ref{t:cho2} shows the percentage of
times that the proposed methodology selects each model and compares it with the other model selection criteria. It is clear
that the traditional criteria have problems to distinguish between models with heavy tails even when the
sample size increases, whilst the proposed methodology performs a robust selection specially for large
sample sizes where tail information is more abundant.

\begin{table}
  \caption{\label{t:cho2} Percentage of the times each model was selected for different sample sizes in Study II.}
  \centering
  \begin{tabular}{c cc cc cc}
    \toprule
    & \multicolumn{3}{c}{Proposed methodology} & \multicolumn{3}{c}{*WAIC}\\
    sample size & Normal & Student-t & Slash & Normal & Student-t & Slash \\
    \cmidrule(lr){1-1}\cmidrule(lr){2-4}\cmidrule(lr){5-7}
    500     & 0\% & 62\%   & 38\% & 6\% & 42\%   & 52\% \\
    1000    & 0\% & 62\%   & 38\% & 0\% & 46\%   & 54\% \\
    2000    & 0\% & 94\%   & 6\%  & 0\% & 58\%   & 42\% \\
    5000    & 0\% & 100\%  & 0\%  & 0\% & 54\%   & 46\% \\
   \bottomrule
   \multicolumn{7}{l}{{\footnotesize *WAIC = all the other criteria (CPO, DIC, EAIC and EBIC) select the same model as WAIC.}}\\
\end{tabular}
\end{table}

Figure~\ref{f:eqm} shows some results regarding the estimation of the regression coefficients. They are quite similar
between the proposed methodology and the other model selection criteria. The same, however, does not happen when we look
at the estimates for the variance $\sigma^2$. The poor performance of the other criteria in selecting the correct model is
clearly reflected in the estimation of the variance, which is significantly overcome by the respective estimates obtained
with the proposed methodology.

\begin{figure}[htb]
  \centering
  {\includegraphics[width=0.35\textwidth]{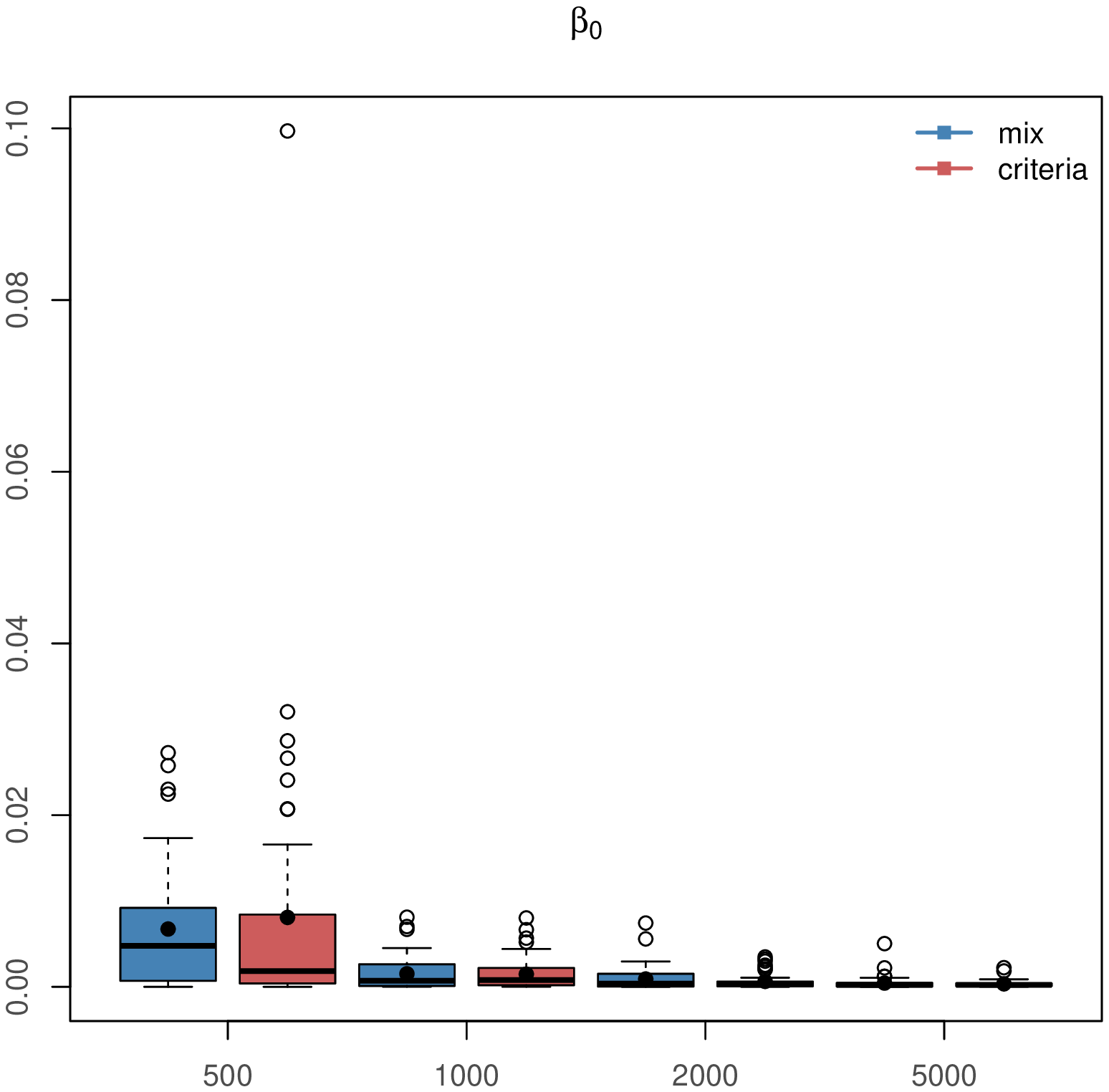}}~{\includegraphics[width=0.35\textwidth]{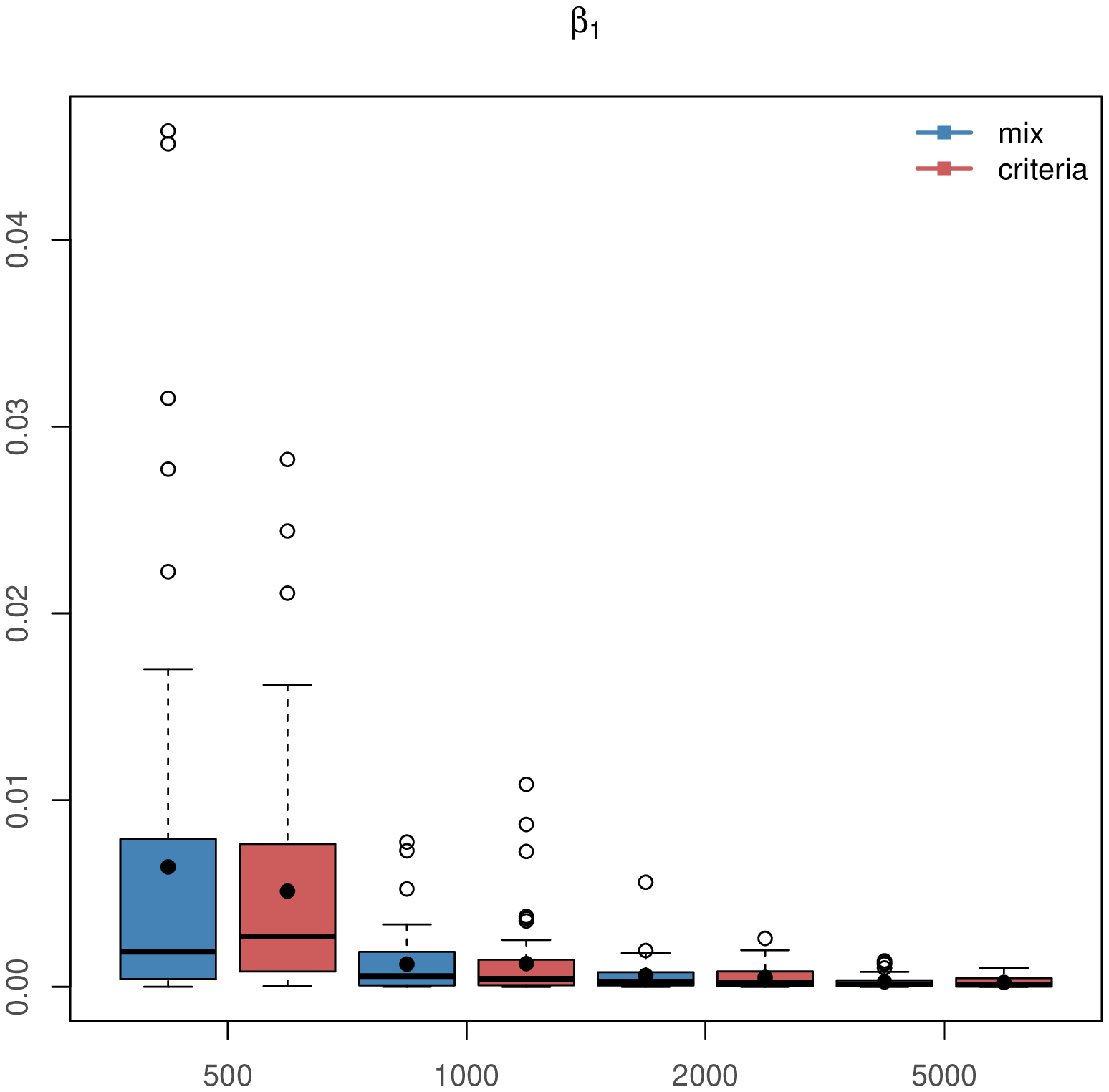}} \\
  {\includegraphics[width=0.35\textwidth]{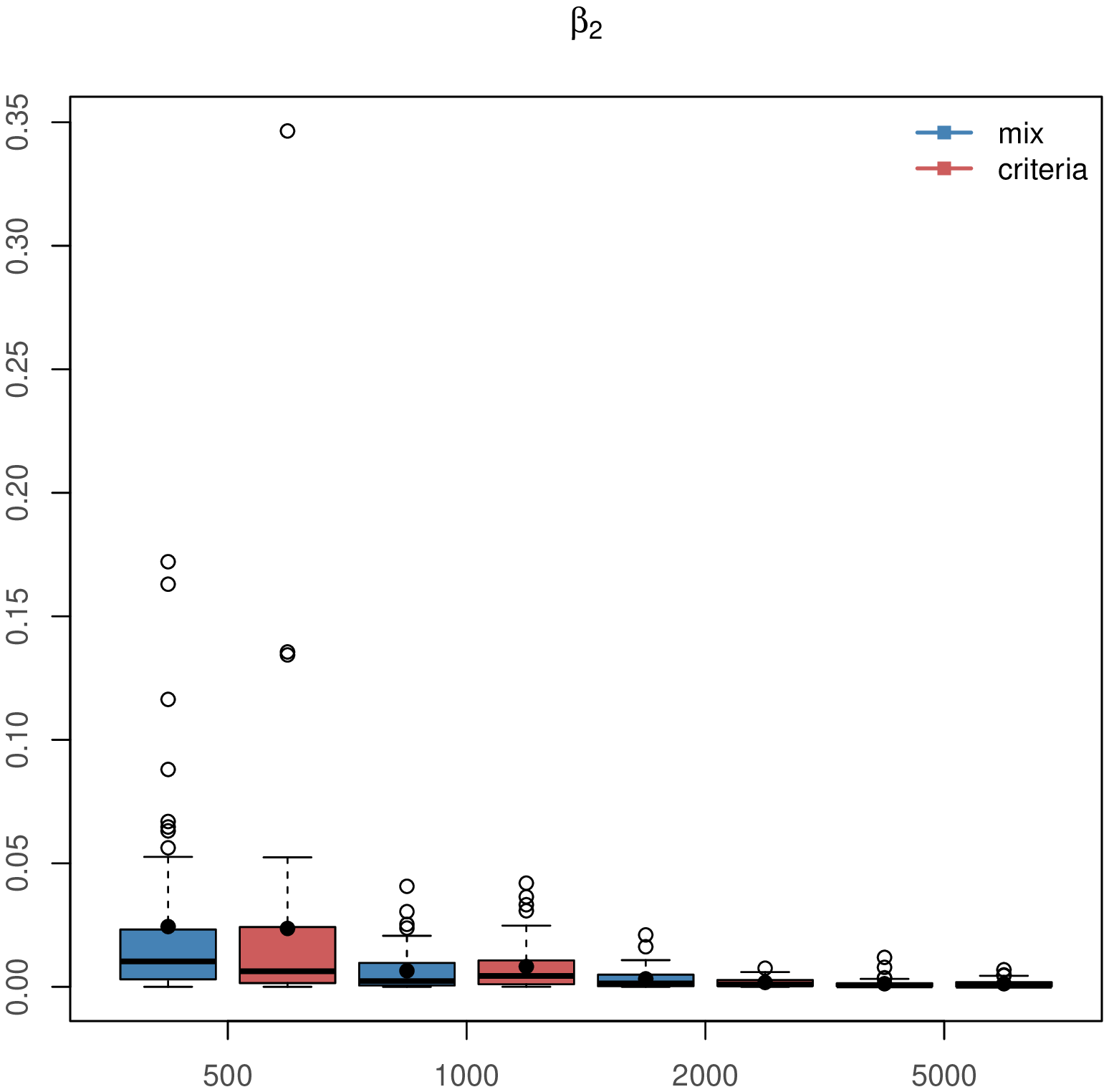}}~{\includegraphics[width=0.35\textwidth]{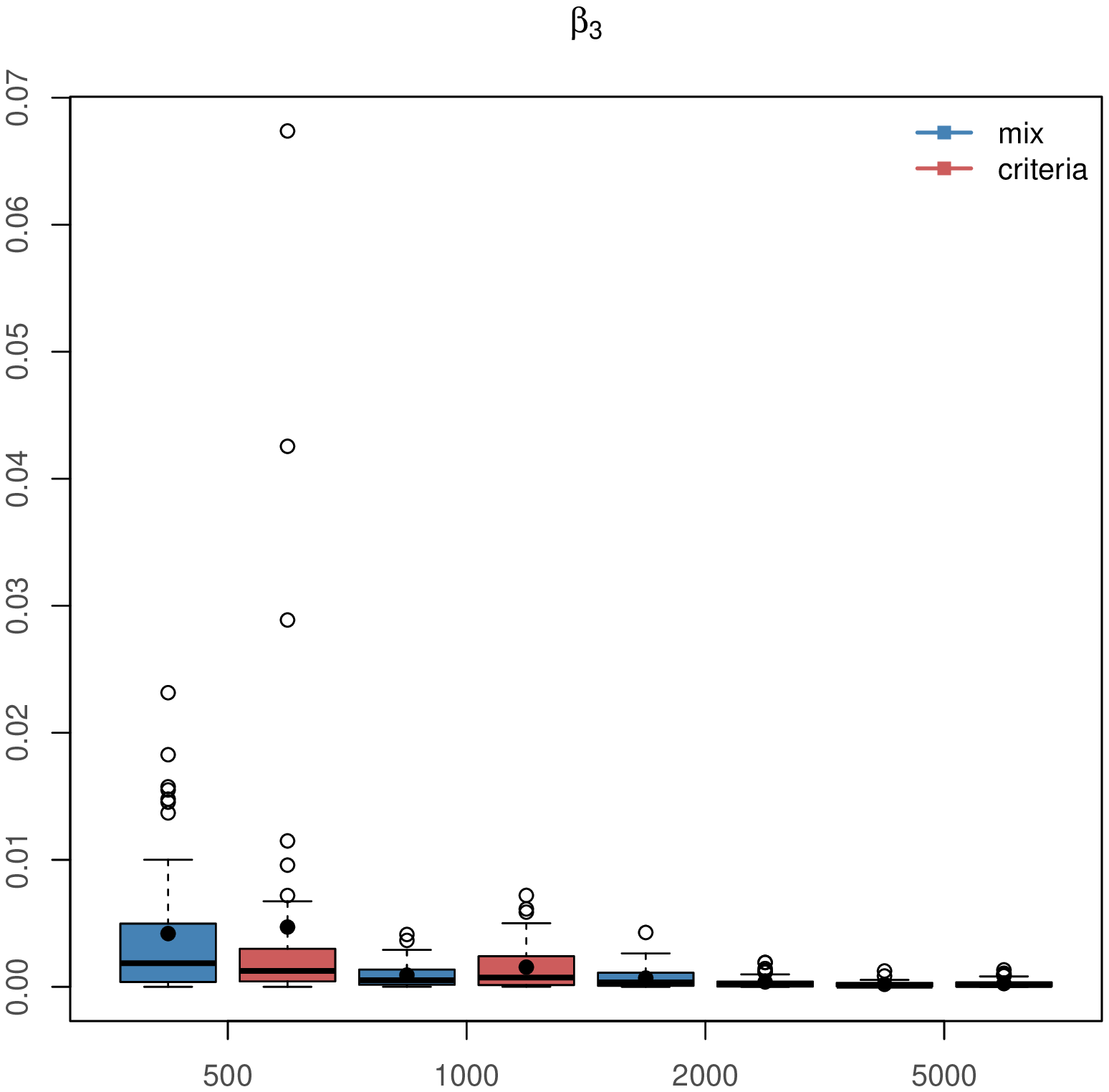}} \\
  \caption{Study II - boxplots of the mean square error (mse) of the $\beta$ estimates (posterior mean under the selected model) for the $50$ replicates for different sample sizes. Colour blue refers to the proposed methodology and red to the other criteria. The black solid dots represent the mean.}
  \label{f:eqm}
\end{figure}

\begin{figure}[htb]
  \centering
  {\includegraphics[width=0.45\textwidth]{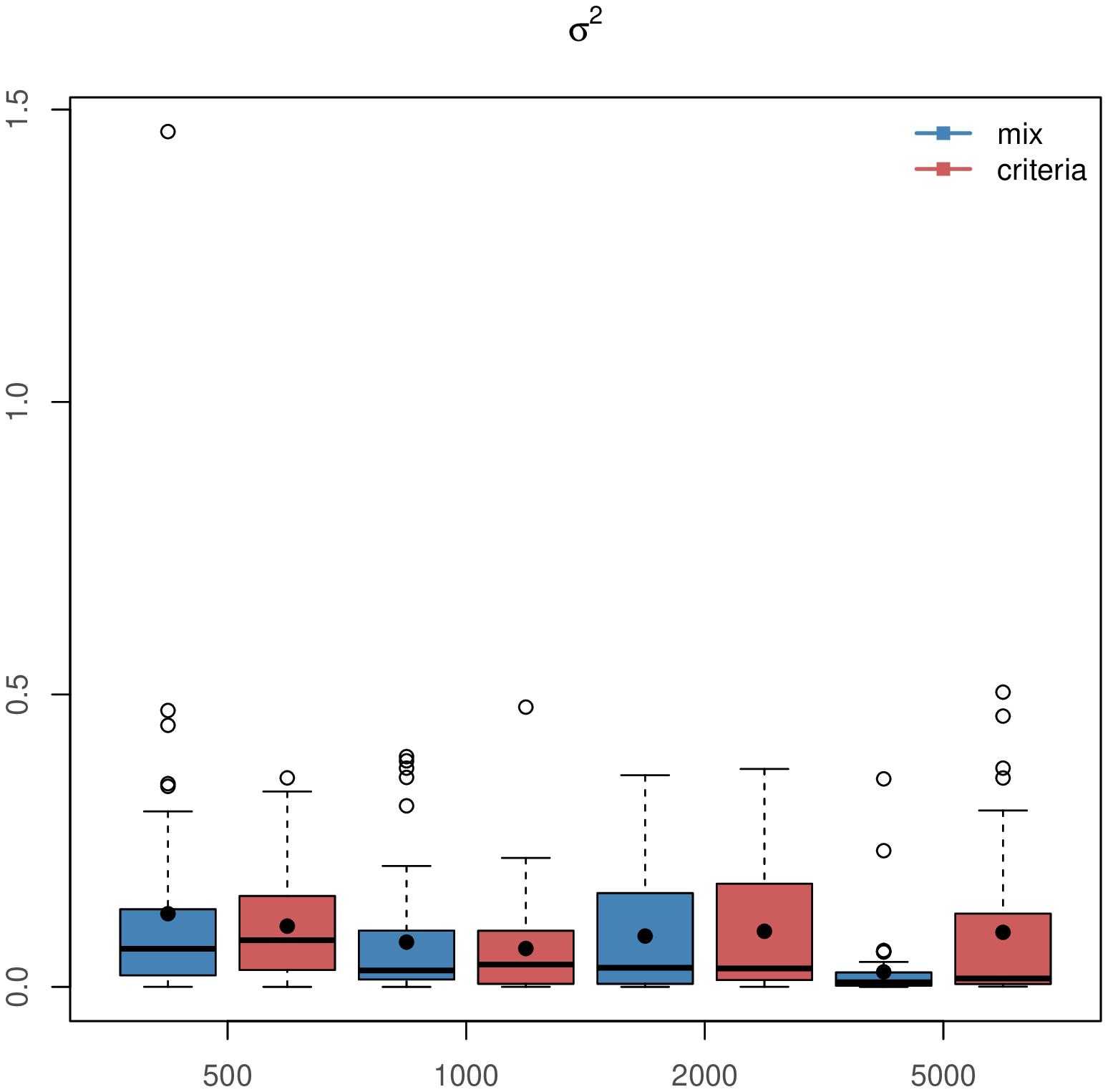}} \\
  \caption{Study II - boxplots of the mean square error (mse) of the $\sigma^2$ estimates (posterior mean under the selected model) for the $50$ replicates for different sample sizes. Colour blue refers to the proposed methodology and red to the other criteria. The black solid dots represent the mean.}
  \label{f:eqm_sig}
\end{figure}

%
%

\subsection{Study III}
\label{secsimIII}

In this study the generating distribution for the
error term is not a specific distribution as in Section~\ref{secsimI} and \ref{secsimII},
but a mixture of the Normal, Student-t and Slash distributions.
More specifically, we consider $e_i \sim 0.1 \mathcal{N}(0,1) + 0.6 \mathcal{T}(0,1,4.00)  + 0.3 \mathcal{S}(0,1,1.15)$,
with the same $\X$, $\bbeta$'s and $\sigma^2$ from Study I. The sample sizes $n$ considered are the same as in Study II.

Again, $50$ replications are generated.
It is important to notice that our modeling framework to perform robust model selection cannot
retrieve the generating model, since we assume that all the residuals
must be from the same distribution. Nevertheless, a good fit may still be provided by one of the individual models.
Table~\ref{t:res2} show the result of one of the $50$ replications. It is clear that the posterior
distribution identifies the Student-t distribution as the best candidate model, specially as the sample size
increases.

\renewcommand{\baselinestretch}{1}

\begin{table}
  \caption{\label{t:res2} Posterior results (mean) for Study III.}
  \centering
  \resizebox{\textwidth}{!}{%
  \begin{tabular}{c cc cc }
    \toprule
    sample size & $\beta^\top = (1, 2, -2$) & $\sigma^2 = 1$ & ($\nu_t$, $\nu_s$) & $\brho = (\rho_1,\rho_2,\rho_3)$\\
    \midrule
    500     & (1.018, 2.014, -1.993) & 1.265  & (3.56, 1.24)  & (0.000, 0.786, 0.214) \\
    1000    & (1.038, 1.972, -1.981) & 0.839  & (4.43, 1.51)  & (0.000, 0.914, 0.086) \\
    2000    & (1.027, 2.012, -2.046) & 0.946  & (4.61, 1.53)  & (0.000, 0.922, 0.078) \\
    5000    & (0.985, 1.979, -2.073) & 0.918  & (4.03, -)     & (0.000, 1.000, 0.000) \\
   \bottomrule
\end{tabular}}
\end{table}


For sample size $2000$, Figure~\ref{f:sep} shows the fit of the selected model (Student-t) and the other two models,
Normal and Slash, fitted individually. It also shows the true generating distribution for the error term.
It is clear that, although the posterior
distribution is different from the true generating distribution, by definition, it approximates fairly very well the
original one.

\begin{figure}[htb]
  \centering
  {\includegraphics[width=0.75\textwidth]{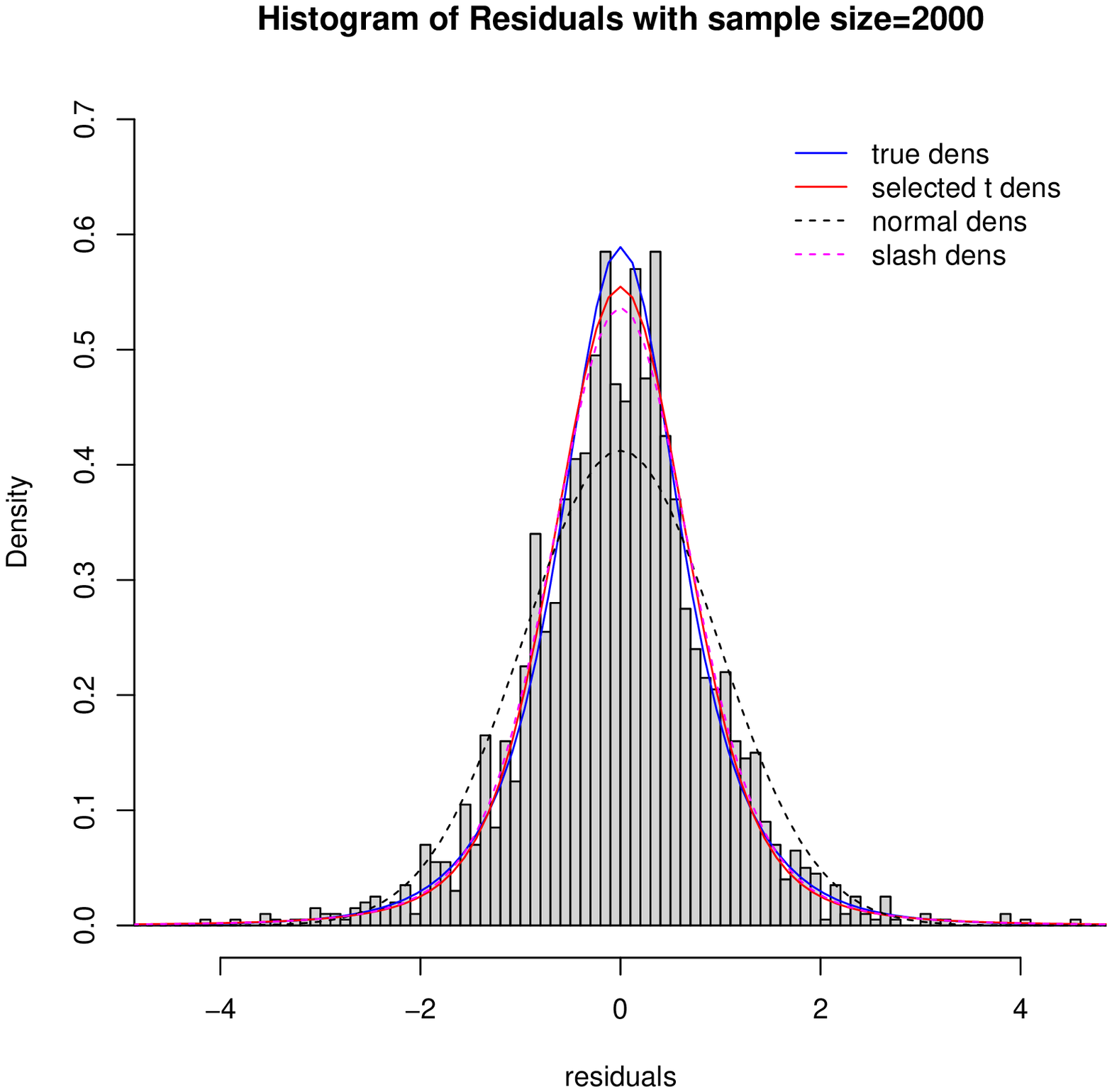}} \\
  \caption{Study III - residual histogram with true generating model (blue), selected Student-t model (red), Normal model (dashed black),
  and Slash model (dashed magenta) for sample size $2000$.}
  \label{f:sep}
\end{figure}


Table~\ref{t:cho3} shows that the proposed model consistently chooses the dominating model (student-t) as
the sample size increases. The same does not happen for the other model selection criteria.

\begin{table}
  \caption{\label{t:cho3} Percentage of the times each model was selected for different sample sizes (Study III).}
  \centering
  \begin{tabular}{c cc cc cc}
    \toprule
    & \multicolumn{3}{c}{Proposed methodology} & \multicolumn{3}{c}{*WAIC}\\
    sample size & Normal & Student-t & Slash & Normal & Student-t & Slash \\
    \cmidrule(lr){1-1}\cmidrule(lr){2-4}\cmidrule(lr){5-7}
    500     & 0\% & 70\%   & 30\% & 6\% & 50\%   & 44\% \\
    1000    & 0\% & 98\%   & 2\%  & 4\% & 46\%   & 50\% \\
    2000    & 0\% & 100\%  & 0\%  & 0\% & 62\%   & 38\% \\
    5000    & 0\% & 100\%  & 0\%  & 0\% & 52\%   & 48\% \\
   \bottomrule
   \multicolumn{7}{l}{{\footnotesize *WAIC = all the other criteria (CPO, DIC, EAIC and EBIC) select the same model as WAIC.}}\\
\end{tabular}
\end{table}

\section{Application}\label{secreal}

\subsection{AIS}\label{secais}
In this section we introduce a biomedical study from the
Australian Institute of Sports (AIS) in $202$ athletes \citep{cook}.
To exemplify our modeling we consider the body mass index (BMI) as our response
and the percentage of body fat (Bfat) as our covariate. This way,
we have a regression model with $\X_{i\cdot} = (1,\mbox{Bfat}_i)$ for $i=1,\ldots,202$.

We fit each individual model separately and the proposed mixture model. Results are presented in Tables \ref{t:crit_ais} and \ref{t:ais}.
Note that, although the Slash model is chosen by all the criteria, and the estimates of the regression coefficients are similar between the individual fit and our model, significant, though not large, differences can be found for the estimates of the variance $\sigma^2$. This highlights the model averaging feature of our approach, which is particularly appealing when one of the models is not chosen with very high probability - in this example $\brho$ - $(0.001, 0.304, 0.695)$.


\renewcommand{\baselinestretch}{1}

\begin{table}
  \caption{\label{t:crit_ais} Model selection criterion for the fitting of the Normal,
  Student-t and Slash regression models.}
  \centering
  \begin{tabular}{cccccc}
    \toprule
    Models & $-$LPML & DIC & EAIC & EBIC & WAIC\\
    \midrule
    Normal     & 498.497 & 2976.407 & 994.142 & 1000.758 & 996.971 \\
    Student-t  & 491.623 & 2935.009 & 982.059 &  991.984 & 983.210 \\
    Slash      & 491.033 & 2931.636 & 980.633 &  990.558 & 982.049 \\
   \bottomrule
\end{tabular}
\end{table}


\renewcommand{\baselinestretch}{1}

\begin{table}
  \caption{\label{t:ais} Posterior results for the BMI analysis with Bfat as covariate
  for the robust mixture model.
  The posterior mean, median a standard deviation (Sd) are presented as well as
  the 95\% high posterior density (HPD) interval.}
  \centering
  \begin{tabular}{cccccc}
    \toprule
    Model & Parameters & Mean & Median & Sd & 95\% HPD interval\\
    \midrule
             & $\beta_0$      & 21.810  & 21.810 & 0.419 & (20.980, 22.620) \\
 Slash Model & $\beta_1$      &  0.070  &  0.070 & 0.028 & (0.015,  0.126)  \\
             & $\sigma^2$     & 10.093  &  8.989 & 3.587 & (5.702, 17.940)  \\
             & $\nu_{s}$     &  1.705  &  1.612 & 0.442 & (1.110, 2.569)   \\
    \midrule
               & $\beta_0$      & 21.794  & 21.799 & 0.418 & (21.022, 22.667) \\
  Slash Selected & $\beta_1$    &  0.071  &  0.071 & 0.028 & (0.016,  0.128)  \\
  Model        & $\sigma^2$     &  9.200  &  8.462 & 2.954 & (5.543, 14.765)  \\
               & $\nu_{s}$   &  1.716  &  1.628 & 0.434 & (1.111, 2.549)   \\
   \bottomrule
\end{tabular}
\end{table}

\subsection{WAGE}\label{secwage}

The wage rate data set presented in \citet{mroz} is used to extend our
modeling framework for censored data. The data consist of the wage
of 753 married white women, with ages between 30 and 60 years old in 1975.
Out of the 753 women considered in this study, 428 worked at some point during
that year. When the wives did not work in 1975, the wage rates were set equal
to zero. However, it is considered that they may had a cost in that year and,
therefore, these observations are considered left censored at zero.
The considered response is $Y_i$ - the wage rate, and the explanatory variables
are the wife's age ($X_{1i}$), years of schooling ($X_{2i}$), number of children younger than six
years old in the household ($X_{3i}$) and number of children between six and nineteen
years old ($X_{4i}$). Thus, $\X_{i\cdot} = (1, X_{1i}, X_{2i}, X_{3i}, X_{4i})$,
$i = 1, \ldots, 753$.

Since the Wage data is censored, we have the following characteristic for our
response variables
\begin{eqnarray} \nonumber
   Y_{obs_i}=\left\{
              \begin{array}{ll}
                \kappa_i, & \mbox{if } Y_i \leq \kappa_i, \\
                Y_i       & \mbox{if } Y_i > \kappa_i,
              \end{array}
            \right.
\end{eqnarray}
with $\kappa_i = 0$.

Suppose that, out of the $n$ responses, $C$ of them are censored as $\kappa_i$. From a Bayesian perspective,
these observations, $\Y_C = (y_1,\ldots,y_C)$, can be viewed as latent and sampled
at each step of the MCMC. Because of the model structure presented
in~(\ref{modeleq1a})-(\ref{modeleq4a}), it is simple to notice that
\begin{equation} \label{eq:cens}
(Y_c|Z_j=1, u_c, \bbeta, \sigma^2, \nu_j) \sim \mathcal{TN}\left(\X \bbeta,\sigma^2\gamma_j u_c^{-1}), \lfloor -\infty, \kappa_c \rfloor \right),\; c=1,\ldots,C,
\end{equation}
where $\mathcal{TN}$ is a truncated Normal distribution with limits $\lfloor -\infty, \kappa_c \rfloor$.
Therefore, we simply add a new sampling step in the blocking scheme as
\begin{equation}
(\p,\Z,\U)\;,\;\Y_C\;,\;\bbeta\;,\;\sigma^2\;,\;(\nu_t,\nu_s).\nonumber
\end{equation}
This simple extension allows our modeling framework to deal with any kind of censored data, where,
for each type of censoring scheme, the new limits of (\ref{eq:cens}) must be calculated.

To obtain our final chain with $100k$ observations, a Markov Chain of $110k$ iterations
is run and the first $10k$ observations are discarded for burn-in.  The posterior
estimate for $\brho$ is $(0.000, 0.025, 0.975)$, which indicates the Slash distribution as the preferred one.

\renewcommand{\baselinestretch}{1}

\begin{table}
  \caption{\label{twage} Posterior results for the Wage data analysis.
  The posterior mean, median a standard deviation (Sd) are presented as well as
  the 95\% high posterior density (HPD) interval.}
  \centering
  \begin{tabular}{ccccc}
    \toprule
    Parameters & Mean & Median & Sd & 95\% HPD interval\\
    \midrule
    $\beta_0$      & -1.174  & -1.152 & 1.408 & (-3.952,  1.523) \\
    $\beta_1$      & -0.109  & -0.108 & 0.022 & (-0.155, -0.066) \\
    $\beta_2$      &  0.646  &  0.645 & 0.070 & (0.508,  0.783)  \\
    $\beta_3$      & -3.114  & -3.103 & 0.387 & (-3.887, -2.381) \\
    $\beta_4$      & -0.293  & -0.294 & 0.129 & (-0.539, -0.039)\\
    $\sigma^2$     & 26.542  & 24.740 & 7.843 & (14.784, 42.624) \\
    $\nu_{s}$      & 1.410   &  1.374 & 0.207 & (1.110, 1.788)   \\
   \bottomrule
\end{tabular}
\end{table}

Table \ref{twage} summarises the posterior results.
\citet{garay2015} studied this data set from a Bayesian perspective fitting
a variety of independent models in the NI family. In their study,
the Slash distribution was selected as the preferred one as in our case.
Moreover, the posterior mean estimates of the fixed effects parameters
in Table \ref{twage} are very similar to the ones presented in
\citet{garay2015} as well as the statistical significance of each covariate.
The wife's age, the number of children younger than six years old in the household
and the number of children between six and nineteen years old tend
to decrease the wage rate, while years of schooling tend to increase the salary.
The posterior estimates encountered by \citet{garay2015} for $\nu_{s}$ and $\nu_t$,
fitting separate models, were $1.438$ and $5.279$, respectively, which agree with our results.
Our mixture model approach was able to correctly
capture the Slash distribution without separately fitting the three models.
Moreover, it provides a high computational gain given the high posterior probability of the Slash model.

\section{Conclusions and some extensions}\label{secconc}

Our proposed methodology has shown considerable flexibility
to perform model selection for heavy-tailed data explained by covariates
under a regression framework. From theoretical arguments,
simulation studies and application to real data sets, it is clear that the
methodology provides a robust alternative to select the
best model instead of relying on model selection criteria
which can be unstable \citep{Gelman}.

In Section~\ref{secwage}
we extend the methodology to censored heavy-tailed
regression, showing that the extension is straightforward and
achieved by adding one simple step to the Gibbs sampler. Also,
the extension of the algorithm described in Section~\ref{s:mcmc}
to include more distributions in the finite mixture is almost direct.
Finally, it is clear from our results that this finite mixture idea can be used
in a variety of problems where a common parametrisation exists for a family
of distributions.

Besides the computational advantage of fitting one general model
instead of $K$ separated models, we also emphasise that our robust model selection framework
automatically performs multiple comparison between the $K$ models,
which gives an advantage if one, instead, prefer to use the Bayes factor
performing 2 by 2 comparisons in each individual model. Moreover, our approach also
allows the use of model averaging.

Although the proposed methodology enriches the class of traditional
censored regression models, we conjecture that the it may not
provide satisfactory result when the response
exhibit asymmetry besides the non-normal behavior. To overcome this
limitation extending the work to account for skewness behavior is
also a possibility, for example by using the scale mixtures of
skew-normal (SMSN) distributions proposed in \citet{lachos}.
Nevertheless, a deeper investigation of those modifications in
the parametrisation and implementations is beyond the scope of
this paper, but provides stimulating topics for further research.
Another possibility of future research is to generalise
these modeling framework to linear mixed model, e.g., clustered,
temporal or spatial dependence. These extensions are being studied
in a different manuscript.

\section*{Appendix}

\appendix

\section{Proof of Lemma 1}\label{a:proof}

The posterior density of $p$ is given by
\begin{equation*}
\ds f(\p|\Y=\y)=\sum_{k=1}^Kf(\p|\Y=\y,Z_k=1)P(Z_k=1|\Y=\y).
\end{equation*}
If we multiply both sides by $p_j$, integrate with respect to $\p$ and use the fact that $\p$ and $\Y$ are conditionally independent given $\Z$, we get
\begin{equation*}
\ds\mathbb{E}[p_j|\y]=\sum_{k=1}^K\mathbb{E}[p_j|Z_k=1]P(Z_k=1|\Y=\y),
\end{equation*}
which is a weighted average of $\{\mathbb{E}[p_j|Z_k=1]\}_{k=1}^K$ and, therefore, implies that
\begin{equation*}
\ds\mathbb{E}[p_j|\y]\in\left(\min_k\left\{\mathbb{E}[p_j|Z_k=1]\right\},\max_k\left\{\mathbb{E}[p_j|Z_k=1]\right\}\right).
\end{equation*}
Now note that $\ds (\p|Z_k=1)\sim Dir\left(\alpha_1+\mathds{1}\{k=1\},\ldots,\alpha_K+\mathds{1}\{k=K\}\right)$ and
$\ds\mathbb{E}[p_j|Z_k=1]$ is $\frac{\alpha_j}{\alpha_0+1}$ if $j\neq k$ and is $\frac{\alpha_j+1}{\alpha_0+1}$ if $j=k$,
where $\alpha_0=\sum_{k=1}^K\alpha_k$. This concludes the proof.

\section{Model Comparison Criteria}\label{a:crit}

The DIC \citep{Spiegelhalter} is a generalisation
of the Akaike information criterion (AIC) and
is based on the posterior mean of the deviance,
which is also a measure of goodness-of-fit.
The $\textrm{DIC}$ is defined by
$$
\textrm{DIC}=\overline{\textrm{D}}(\bftheta)+\rho_{\textrm{D}} = 2
\overline{\textrm{D}}(\bftheta) - \textrm{D}(\tilde{\bftheta}),
$$
where $\tilde{\bftheta} = \textrm{E}[\bftheta|\y]$,
$\overline{\textrm{D}}(\bftheta)$ is the posterior expectation of the deviance
and $\rho_{\textrm{D}}$ is a measure of the effective number of parameters in the model.
The effective number of parameters, $\rho_{\textrm{D}}$, is defined as $\rho_{\textrm{D}}= \overline{\textrm{D}}(\bftheta) -
\textrm{D}(\tilde{\bftheta})$, with $\overline{\textrm{D}}(\bftheta)=-2\textrm{E}[ \log f(\y|\bftheta) | \y]$.

The computation of the integral $\overline{\textrm{D}}(\bftheta)$ is complex,
a good solution can be obtained using the MCMC sample
$\{\bftheta_{1},\dots, \bftheta_{M}\}$ from the posterior
distribution. Thus,  we can obtain an approximation of the
$\textrm{DIC}$ by first computing the sample posterior mean of the
deviations
$\overline{\textrm{D}}= -2 \frac{1}{M} \sum_{m=1}^{M} \log f(\y|\bftheta_m)$ and then
$\widehat{\textrm{DIC}}= 2\overline{\textrm{D}} - \textrm{D}(\tilde{\bftheta})$.

The expected Akaike information criterion (EAIC),
and the expected Bayesian information criterion (EBIC)
\citep[see discussion at][]{Spiegelhalter} are given by
$$\widehat{EAIC}=\overline{\textrm{D}}+2\vartheta~~~ \mbox{and}~~~ \widehat{EBIC}=\overline{\textrm{D}}+\vartheta\log\left(n\right),  $$
respectively, where $\vartheta$ is the number of model parameters and
can be used for model comparison.

Recently, \citet{Watanabe} introduced the Widely Applicable Information Criterion (WAIC).
The WAIC is a fully Bayesian approach for estimating the out-of-sample expectation.
The idea is to compute the log pointwise posterior predictive density ($\it{lppd}$) given by
$ lppd = \sum_{i=1}^n \log \left( \frac{1}{M} \sum_{m=1}^{M} f(y_i|\bftheta_m) \right)$,
and then, to adjust for overfitting, add a term to correct for effective number of parameters
$\rho_{\textrm{WAIC}} = \sum_{i=1}^n V_{m=1}^M(\log f(y_i|\bftheta_m))$,
where $V_{m=1}^M(a) = \frac{1}{M-1} \sum_{m=1}^{M} ( a_m - \bar{a})^2$. Finally, as
proposed by \citet{Gelman}, the WAIC is given by
$$
\textrm{WAIC} = -2 (lppd - \rho_{\textrm{WAIC}}).
$$
So far, for the DIC, EAIC, EBIC and WAIC, the model that best fits a data set
is the model with the smallest value of the criterion.

Another common alternative is the conditional predictive ordinate ($CPO$) approach \citep{Geisser}.
This statistic is based on the cross validation criterion to compare the models.
Let $\y= \left\{y_1,\cdots,y_n\right\}$ be an observed sample from
$f\left(\cdot|\bftheta\right)$. For the $i$-th observation, the $CPO_i$ can be written as:
\begin{eqnarray}\label{CPO} \nonumber
CPO_i = p\left(y_i|\y_{(-i)}\right) = \int_{\bftheta\in
\bfTheta}f\left(y_i|\bftheta \right)\pi\left(\bftheta |\y_{(-i)}\right)d\bftheta
= \left\{\int_{\bftheta\in \bfTheta}\frac{\pi\left(\bftheta
|\y\right)}{f\left(y_i|\bftheta \right)}d\bftheta \right\}^{-1},
\end{eqnarray}
where $\y_{(-i)}$ is the $\y$ without the $i$-th observation
and $\pi\left(\bftheta |\y\right)$ denotes the posterior distribution of $\bftheta$.
Thus, the $CPO_i$ has the idea of the leave one out cross validation, where each value is
an indicator of the likelihood value given all the other observations. For this reason, low
values of $CPO_i$ must correspond to poorly fitted observations.
For many models, the analytic calculation of the $CPO$ is not
available. However, \citet{Dey} showed that an harmonic mean approach
can be used to do a Monte Carlo approximation of the $CPO_i$ by using a MCMC sample
$\left\{\bftheta_1,\cdots,\bftheta_M\right\}$ from the posterior
distribution $\pi\left(\bftheta |\y\right)$. Therefore,
the $CPO_i$ approximation is given by
\begin{eqnarray*}
\widehat{CPO}_i = \left\{\frac{1}{M}\sum^{M}_{m=1}\frac{1}{f\left(y_i|\bftheta_m \right)}\right\}^{-1}.
\end{eqnarray*}
Since the $CPO_i$ is defined for each observation, the log-marginal pseudo
likelihood (LPML) given as
$$\textrm{LMPL}=\sum^{n}_{i=1}\log\left(\widehat{CPO}_i\right),$$
is used to summarise the $CPO_i$ information and the larger the value
of LMPL is, the better the fit of the model under consideration.

\bibliographystyle{spbasic}
\bibliography{biblio}

\end{document}